\documentclass[11pt]{article}
\usepackage{amsmath,amssymb,amsfonts}
\usepackage{latexsym}
\usepackage{graphicx}
\usepackage{algorithm,algpseudocode}
\usepackage{times}
\renewcommand\thefootnote{\arabic{footnote}}

\usepackage{amsthm}
\newtheorem{theorem}{Theorem}[section]
\newtheorem{lemma}[theorem]{Lemma}
\newtheorem{proposition}[theorem]{Proposition}

\title{
Any Finite Distributive Lattice is Isomorphic to\\
the Minimizer Set of an ${\rm M}^{\natural}$-Concave Set Function
 }
\author{
 Tomohito Fujii\footnotemark[1] \and 
 Shuji Kijima\footnotemark[1]~\footnotemark[2]
 }

\begin{document}
\maketitle
\renewcommand\thefootnote{\fnsymbol{footnote}}
\footnotetext[1]{
 Graduate School of Information Science and Electronic Engineering, 
 Kyushu University
}
\footnotetext[2]{
JST PRESTO, 744 Motooka, Nishi-ku, Fukuoka, 819-0395, Japan
}

\renewcommand\thefootnote{\arabic{footnote}}

\begin{abstract}
 Submodularity is an important concept in combinatorial optimization, and 
  it is often regarded as a discrete analog of convexity. 
 It is a fundamental fact that the set of minimizers of any submodular function forms a distributive lattice. 
 Conversely, it is also known that  
  any finite distributive lattice is isomorphic to the minimizer set of a submodular function, 
  through the celebrated Birkhoff's representation theorem. 
 ${\rm M}^{\natural}$-concavity is a key concept in discrete convex analysis. 
 It is known for set functions 
  that the class of ${\rm M}^{\natural}$-concavity 
   is a proper subclass of submodularity. 
 Thus, the minimizer set of an ${\rm M}^{\natural}$-concave function forms a distributive lattice. 
 It is natural to ask if any finite distributive lattice appears 
   as the minimizer set of an ${\rm M}^{\natural}$-concave function.  
 This paper affirmatively answers the question. 

\bigskip
\noindent
{\bf Keywords: }
 ${\rm M}^{\natural}$-concave, valuated matroid, submodular, distributive lattice, {\#}BIS
\end{abstract}

\section{Introduction}
\subsection{The minimizer set of a submodular function forms a distributive lattice}
 A set function $f \colon 2^N \to \mathbb{R}$ for a finite set $N$ is {\em submodular} if 
\begin{align}
 f(X) + f(Y) \geq f(X \cup Y) + f(X \cap Y)
\end{align} 
 holds for any $X,Y \in 2^N$. 
 Submodularity is often regarded as a discrete analog of convexity~\cite{Lovasz83,FujishigeB,MurotaB,NYKY17+}. 
 A submodular function is efficiently minimized~\cite{Schrijver00,IFF01,IO09,LSC15}, and 
  it has many applications in economics, machine learning, etc.  
 On the other hand, 
  maximization of a submodular function, e.g., max cut, is a celebrated NP-hard problem, and 
  approximation is  recently investigated with applications in machine learning, see e.g.,~\cite{NWF78,FMV07}. 

 It is a fundamental fact on submodular functions 
   that the set of minimizers of a submodular function forms a distributive lattice. 
 Conversely, any finite distributive lattice ``appears'' as the minimizers of a submodular function. 
 For a finite {\em partially ordered set} ({\em poset}) $\mathcal{P}=(N,\preccurlyeq)$, 
   $I \subseteq N$ is an {\em ideal} of $\mathcal{P}$
   if $x \preccurlyeq y \in I \Rightarrow x \in I$ holds for any $x, y\in N$.
 Let ${\cal I}({\cal P})$ denote the set of whole ideals of the poset ${\cal P}$. 
 Then, ${\cal I}({\cal P})$ forms a distributive lattice. 
 The following celebrated theorem is due to Birkhoff~\cite{Birkhoff37}. 
\begin{theorem}[Birkhoff's representation theorem \cite{Birkhoff37,FujishigeB}]
\label{thm:Bikhoff}
 For any finite distributive lattice ${\cal D}$, 
   there exists a poset ${\cal P}$ such that ${\cal I}({\cal P})$ is isomorphic to ${\cal D}$. 
\end{theorem}

 Using Theorem~\ref{thm:Bikhoff}, 
  it is known that 
 any finite distributive lattice is isomorphic to the minimizer set of a submodular function, as follows.  
\begin{proposition}[see e.g., \cite{FujishigeB}]\label{prop:down-submo}
 As given a finite poset ${\cal P}=(N,\preccurlyeq)$, let $f \colon 2^N \to \mathbb{R}$ be defined by 
\begin{align}
 f(X) &= |\{ j \in N \setminus X : \mbox{$\exists i \in X$ such that $j \prec i$}  \}|
\label{def:down-submo}
\end{align}
 for any $X \in 2^N$, 
  where $j \prec i$ denotes $j \preccurlyeq i$ and $j \neq i$. 
 Then $f$ is submodular, and it satisfies 
\begin{equation*}
f(X)\begin{cases}
=0& \text{if $X \in {\cal I}({\cal P})$,}\\
>0& \text{otherwise,}
\end{cases}
\end{equation*}
 for any $X \in 2^N$. 
\end{proposition}

 Proposition~\ref{prop:down-submo} provides 
  a representation of a finite distributive lattice with a submodular function; 
  any finite distributive lattice is represented as the minimizer set of a submodular function. 
 Another interesting representation theorem for finite distributive lattices is described by {\em stable matchings}; 
 John Conway showed that the set of stable matchings forms 
  a distributive lattice under the preferences of Men (or Women, similarly)~\cite{Knuth91}. 
 Blair~\cite{Blair} showed that 
  any finite distributive lattice is isomorphic to 
  the distributive lattice of the stable matchings for a stable marriage instance.  

\subsection{${\rm M}^{\natural}$-concavity is a proper subclass of submodularity}
 A set function $f\colon 2^N\rightarrow \mathbb{R}$ is {\em ${\rm M}^\natural$-concave} (cf.~\cite{Murota18}) 
 if, for any $X,Y\in 2^N$ and $i \in X\setminus Y$, we have 
\begin{equation}
  f(X)+f(Y) \leq f(X-i)+f(Y+i)
\label{mconcave1}
\end{equation}
 or else  
\begin{equation}
f(X)+f(Y) \leq f(X-i+j)+f(Y+i-j)
\label{mconcave2}
\end{equation}
 holds for some $j \in Y \setminus X$, 
 where $X-i, Y+i, X-i+j, Y+i-j$ are abbreviations of 
  $X \setminus \{ i \}, Y \cup \{ i \}, (X \setminus \{ i \})\cup \{ j \}, (Y \cup \{ i \})\setminus \{ j \}$, respectively.
 ${\rm M}^{\natural}$-concavity is introduced by Murota~\cite{MurotaB,MS99} 
  as a quantitative version of matroid extending the exchange property to set functions~\cite{Murota18}, and  
  it is a closely related to {\em valuated matroid} introduced by Dress and Wenzel~\cite{DW90,DW92}. 
 Fujishige and Yang \cite{FY03} showed that ${\rm M}^{\natural}$-concavity is equivalent 
  to {\em gross substitutes property} of Kelso and Crawford~\cite{KC82} in economics. 
 ${\rm M}^{\natural}$-concavity is also found in many areas such as 
   systems analysis,
   inventory theory in operations research, and 
   mathematical economics and game theory including stable matching~\cite{Murota18,MurotaB00,ST15,FT07}. 
 ${\rm M}^{\natural}$-concavity is extensionally defined on multidimensional integer lattice, and 
  it is a key concept in the theory of ``discrete convex analysis''
  \cite{MurotaB,Murota10,Murota16,Murota18,MS99,Shioura98,ST15}. 

 Interestingly, any ${\rm M}^{\natural}$-concave set function is submodular~\cite{MurotaB}. 
 Thus, 
  an ${\rm M}^{\natural}$-concave set function is minimized efficiently, 
    using an algorithm for submodular minimization. 
 In contrast, 
   not every submodular function is ${\rm M}^{\natural}$-concave; 
   meaning that ${\rm M}^{\natural}$-concavity is a proper subclass of submodularity for set functions. 
 In fact, any ${\rm M}^{\natural}$-concave set function is efficiently maximized by a greedy algorithm, 
  so is a matroid rank function \cite{DW92,Shioura98,Murota10}. 
 It is known that 
  the set of maximizers of an ${\rm M}^{\natural}$-concave set function forms 
   a {\em generalized matroid} (a.k.a.\ ${\rm M}^{\natural}$-convex family), and conversely 
  any generalized matroid appears as the maximizer set of an ${\rm M}^{\natural}$-concave set function~\cite{MurotaB,Murota16}.

 Since an ${\rm M}^{\natural}$-concave set function is submodular, 
   the minimizer set of an ${\rm M}^{\natural}$-concave set function forms a distributive lattice. 
 It is a natural question if 
   any finite distributive lattice appears as the minimizer set of an ${\rm M}^{\natural}$-concave function. 
 A naive candidate may be the function given by \eqref{def:down-submo}. 
 We briefly remark that this is not the case.

\begin{proposition}\label{prop:bad-example}
 The set function $f$ given by \eqref{def:down-submo} is NOT ${\rm M}^{\natural}$-concave, in general. 
\end{proposition}
\begin{figure}
\centering
\includegraphics[width=4cm,clip]{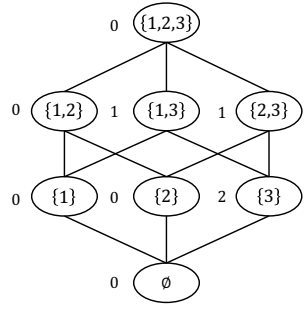}
\caption{The set function $f$ in the proof of Proposition~\ref{prop:bad-example}}
\label{zu1}
\end{figure}
\begin{proof}
 Let ${\cal P} = (\{1,2,3\}, \preccurlyeq)$ be given by $1 \prec 3$ and $2 \prec 3$.
 Then, ${\cal I}({\cal P}) = \{ \emptyset,\{1\}, \{2\}, \{1,2\},\{1,2,3\} \}$.
 We will check that 
  the set function $f$ given by \eqref{def:down-submo} 
   is not M$^{\natural}$-concave for ${\cal P}$ (see Figure~\ref{zu1}).
 Let $X=\{ 3 \}$, $Y=\{ 1,2 \}$ and $i=3$. 
 Then, 
  it is easy to observe that $f$ does not satisfy \eqref{mconcave1}, 
  because $f(X)=2$, $f(Y)=0$, $f(X-i)=0$ and $f(Y+i)=0$. 
 It is also easily confirmed that $f$ does not satisfy \eqref{mconcave2} for any $j \in \{1,2\}$.
 For $j=1$, 
  $f(X-i+j)=f(\{1\})=0$ and $f(Y+i-j)=f(\{2,3\})=1$, then \eqref{mconcave2} does not hold.
 For $j=2$, 
  $f(X-i+j)=f(\{2\})=0$ and $f(Y+i-j)=f(\{1,3\})=1$, then \eqref{mconcave2} does not hold.
\end{proof}

 This paper presents a representation theorem for finite distributive lattices, and 
   affirmatively answers the question. 
 Our result supports the fact that 
   ${\rm M}^{\natural}$-concavity covers a large part of submodularity for set functions; 
   the minimizer sets of ${\rm M}^{\natural}$-concave set functions represent all finite distributive lattices, 
   so do submodular functions.

\section{$\mathrm{M}^\natural$-concave functions with the minimizer set ${\cal I}({\cal P})$}
This section establishes the following theorem. 
\begin{theorem}\label{thm:main}
 For any finite poset ${\cal P} = (N,\preccurlyeq)$, 
  the exists an $\mathrm{M}^\natural$-concave function $f\colon 2^N \to \mathbb{R}$ satisfying 
\begin{equation}\label{cond:min}
f(X)\begin{cases}
=0& \text{if $X \in {\cal I}({\cal P})$,}\\
>0& \text{otherwise,}
\end{cases}
\end{equation}
 for any $X \in 2^N$. 
\end{theorem}
 In proofs of Theorem~\ref{thm:main},  
 we will use the following known fact (see Appendix for a proof). 
\begin{proposition}[cf. \cite{Murota16} (Section 3.6 ``Example'' 6, p.~173)]\label{prop:Murota}
 Let $G=(U,V;E)$ be a bipartite graph with vertex bipartition $(U,V)$ 	and edge set $E$, and  
 let $w\colon E \to \mathbb{R}$ be an edge weight function. 
 For $M \subseteq E$, we denote by $\partial M$ the set of  the vertices incident to some edge in $M$, and 
 call $M$ a matching  if $|U \cap \partial M| = |M| = |V \cap \partial M|$. 
 For $X \subseteq U$ denote by $f(X)$ the maximum weight of a matching that precisely matched $X$ in $U$, i.e.,  
\begin{equation}
 f(X) = \max\{w(M) : \mbox{$M \subseteq E$ is a matching satisfying $\partial M \cap U = X$}\}
\end{equation}
 with $w(M)=\sum_{e \in M}w(e)$, where $f(X)=-\infty$\footnote{
   Here we omit the argument on ``effective domain'' concerning $f(X)=-\infty$ (see e.g., \cite{Murota16}). 
   The ${\rm M}^{\natural}$-concave functions appearing in this paper satisfy $f(X)>-\infty$ for any $X \in 2^N$.
 } if no such $M$ exists for $X$. 
 Then $f\colon 2^U \to \mathbb{R} \cup \{-\infty\} $ is an $M^{\natural}$-concave function.  
\end{proposition}

\subsection{A simple ${\rm M}^{\natural}$-concave set function}
 Now, we present an ${\rm M}^{\natural}$-concave set function for Theorem~\ref{thm:main}. 
 As given an arbitrary finite poset ${\cal P}=(N,\preccurlyeq)$ of order $n$, 
  let $G_0 =(U,V;E_0)$ be a bipartite graph with 
  vertex set consisting of 
   $U=\{u_1,u_2,\dots,u_n\}$ and $V=\{v_1,v_2,\dots,v_n\}$, and 
  edge set given by 
   $E_0=\{\{u_i,v_j\}:\mbox{$u_i\in U$, $v_j\in V$ and $j\preccurlyeq i$ on $\mathcal{P}$}\}$. 
  Let $w_0 \colon E_0 \to \mathbb{Z}_{\geq 0}$ be the edge weight function given by 
\begin{align}
 w_0(\{u_i,v_j\})= \begin{cases}
  0 &\mbox{if $i=j$,}\\
  1 &\mbox{otherwise,}
 \end{cases} 
\label{def:w-simple}
\end{align}
  for any edge $\{u_i,v_j\}\in E_0$. 
 Let $f_0\colon 2^N\rightarrow \mathbb{R}$ be defined by
\begin{equation}\label{def:f-simple}
 f_0(X) = \max{\{w_0(M): 
   \mbox{$M \subseteq E_0$ is a matching satisfying $\partial M \cap U = U_X$} \}
   } 
\end{equation}
 for any $X \in 2^N$, where $w_0(M) = \sum_{e \in M}w_0(e)$ and $U_X = \{u_i \in U : i \in X\}$. 

\begin{lemma}\label{prop:simple}
 For any finite poset, $f_0$ is $\mathrm{M}^\natural$-concave, and it satisfies \eqref{cond:min}. 
\end{lemma}
\begin{proof}
 The ${\rm M}^{\natural}$-concavity of $f_0$ follows from Proposition~\ref{prop:Murota}. 
 \eqref{cond:min} is not difficult, as follows. 
 If $X \not\in {\cal I}(P)$, 
  there exist $i \in X$ and $j \not\in X$ such that $j\prec i$. 
 Then, there is a matching $M$ 
   such that $\{u_i,v_j\} \in M$ and $\partial M \cap U= U_X$. 
 Since $w(\{u_i,v_j\})=1$, $f(X) \geq w(M) > 0$, and we obtain \eqref{cond:min} in this case.  
 Suppose $X \in {\cal I}(P)$. 
 Then, $\{\{u_i,v_i\} : i \in X \}$ is a trivial matching, and its weight is zero. 
 We claim that the trivial matching is the unique matching satisfying $\partial M \cap U= U_X$ in the case. 
 Let $\emptyset = Y_0 \subset Y_1  \subset \cdots \subset Y_k = X$ be 
  a maximal chain from $\emptyset$ to $X$ on ${\cal I}({\cal P})$ where $k=|X|$. 
 Notice that 
  $|\Gamma(U_{Y_l}\})| = l$ holds for $l=0,1,\ldots, k$ since $Y_l \in {\cal I}({\cal P})$, 
 where $\Gamma(U')$ for $U' \subseteq U$ denotes the adjacent vertices of $U'$, 
 i.e.,  $\Gamma(U') = \{ v \in V : \exists u \in U',\ \{u,v\}\in E_0 \}$. 
 By an induction on $l$, we can see that 
  only the trivial matching for $Y_l$ satisfies the condition $\partial M \cap U= U_{Y_l}$. 
 Thus, $f_0(X)=0$ in the case. 
\end{proof}
 Theorem~\ref{thm:main} follows Lemma~\ref{prop:simple}.

\subsection{Another $\mathrm{M}^\natural$-concave set function}
 This subsection presents another $\mathrm{M}^\natural$-concave function for Theorem~\ref{thm:main},  
  in fact it is presented in a preliminary version of this manuscript. 
  Here, we give a simpler proof. 
 As given an arbitrary finite poset ${\cal P}=(N,\preccurlyeq)$ of order $n$, 
   let $G_1 =(U,V;E_1)$ be a bipartite graph 
    with vertex set consisting of 
    $U=\{u_1,u_2,\dots,u_n\}$ and $V=\{v_1,v_2,\dots,v_n\}$, and 
   edge set given by 
    $E_1=\{\{u_i,v_j\}:\mbox{$u_i\in U$, $v_j\in V$ and $j\prec i$ on $\mathcal{P}$}\}$.  
 Note that $E_1 = E_0 \setminus \{\{u_i,v_i\} : i \in N\}$. 
  Let $w_1 \colon E_1 \to \mathbb{Z}_{\geq 0}$ be the edge weight function given by 
\begin{align}
 w_1(\{u_i,v_j\})=\max\{ |S|-1: \text{$S \subseteq N$ is a chain such that $j \preccurlyeq s \preccurlyeq i$ for any $s \in S$} \}, 
\label{def:w1}
\end{align}
  for any edge $\{u_i,v_j\}\in E_1$, 
 where $S \subseteq N$ is a {\em chain} of $S$ if $(S,\preccurlyeq)$ is a totally ordered set, 
	  i.e., $w_1(\{u_i,v_j\})$ denotes the ``length'' of a maximum  chain between $j$ and $i$ for $j \prec i$.  
 Let $f_1\colon 2^N\rightarrow \mathbb{R}$ be defined by
\begin{equation}\label{def:f1}
 f_1(X) = \max{\{w_1(M): 
 \mbox{ $M \subseteq E_1$ is a matching satisfying $\partial M \subseteq U_X \cup V_{\overline{X}}$}\} } 
\end{equation}
 for any $X \in 2^N$, 
  where $U_X = \{ u_i \in U : i\in X \}$ and $V_{\overline{X}} = \{ v_i \in V: i \not\in X \}$.

\begin{lemma}\label{teiri1}
 For any finite poset, 
  $f_1$ is $\mathrm{M}^\natural$-concave, and it satisfies \eqref{cond:min}. 
\end{lemma}

 It is not trivial from Proposition~\ref{prop:Murota} that $f_1$ is ${\rm M}^{\natural}$-concave\footnote{
   We gave a naive proof in the preliminary version.
   }. 
 To prove Lemma~\ref{teiri1}, we introduce another set function $f_2$, as follows. 
 As given an arbitrary finite poset ${\cal P}=(N,\preccurlyeq)$ of order $n$, 
  let $G_2 = (U,V;E_2)$ be given by $G_2=G_0$. 
 Let $w_2 \colon E_2 \to \mathbb{Z}_{\geq 0}$ be the edge weight function given by 
\begin{align}
 w_2(\{u_i,v_j\})= \begin{cases} 
 0 & \mbox{if $i=j$,}\\
 w_1(\{u_i,v_j\})& \mbox{otherwise,} 
 \end{cases}
\label{def:w2}
\end{align}
  for any edge $\{u_i,v_j\}\in E_2$. 
 Let $f_2 \colon 2^N\rightarrow \mathbb{R}$ be defined by
\begin{equation}\label{def:f2}
 f_2(X) = \max{\{w_2(M): 
   \mbox{$M \subseteq E_2$ is a matching satisfying  
   $\partial M \cap U = X$} \}
   } 
\end{equation}
 for any $X \in 2^N$, where $U_X=\{ u_i \in U : i\in X \}$.

\begin{lemma}\label{lem:f1=f2}
 For any finite poset, $f_2 \equiv f_1$. 
 Furthermore, $f_2$ is $\mathrm{M}^\natural$-concave, and it satisfies \eqref{cond:min}. 
\end{lemma}
\begin{proof}
 Firstly, we prove $f_1(X) \leq f_2(X)$ for any $X \in 2^N$. 
 Suppose that $M \subseteq E_1$ is a matching attaining  $f_1(X)=w_1(M)$, i.e., $\partial M \subseteq U_X \cup V_{\overline{X}}$. 
 Let 
  $M'=M\cup \{\{u_i,v_i\} : u_i \in  U_X \setminus \partial M\}$.
 Then, $\partial M'\cap U =U_X$ holds, and hence 
  $f_2(X) \geq w_2(M') = w_1(M) = f_1(X)$.

 Next, we prove $f_1(X) \geq f_2(X)$ for any $X \in 2^N$. 
 Suppose that $M \subseteq E_2$ is a matching attaining  $f_2(X)=w_2(M)$, i.e., $\partial M \cap U = U_X$. 
 We iteratively construct a matching $M' \subseteq E_1$ satisfying $\partial M' \subseteq U_X \cup V_{\overline{X}}$, as follows.  
 To begin with, 
   set $M'_1 = M \setminus \{\{u_i,v_i\} : i \in X\}$. 
 If $\partial M'_1 \cap V_X = \emptyset$ then $M'_1$ is a desired matching, and 
   we obtain $f_1(X) \geq f_2(X)$ in the case, since $w_1(M) = w_2(M'_1)$, clearly. 
 Suppose $v_i \in \partial M'_1 \cap V_X$.  
 Here, we remark that $M'_1$ satisfies the condition (*) 
   ``if $v_i \in \partial M'_1 \cap V_X$ then $u_i  \in \partial M'_1$,'' since $\partial M \cap U = U_X$.  
 Without loss of generality, we may assume that 
    $\{u_k,v_i\} \in M'_1$ and $\{u_i,v_j \} \in M'_1$. 
 This implies that $j \prec i \prec k$, and hence $\{u_k,v_j\} \in E_2$. 
 Set $M'_2:= M'_1 \cup \{\{u_k,v_j\}\} \setminus \{\{u_k,v_i\},\{u_i,v_j\}\}$. 
 Then, we obtain a matching $M'_2$, 
 such that $\partial M'_2 \cap V_X = (\partial M'_1 \cap V_X) \setminus \{v_i\}$ holds and 
   $M'_2$ inherits the condition (*). 
 We also remark that $w_2(M'_2) \geq w_2(M'_1)$ holds, 
   since $w_2(\{u_k,v_j\}) \geq  w_2(\{u_k,v_i\}) + w_2(\{u_i,v_j\})$ by \eqref{def:w1}. 
 Recursively applying the above arguments, 
   we eventually obtain a matching $M'$ such that $\partial M' \cap V_X = \emptyset$, 
   meaning that  $\partial M' \subseteq U_X \cup V_{\overline{X}}$. 
 It is not difficult to see from the above argument, $f_1(X) \geq w_1(M') \geq w_2(M) = f_2(X)$. 

 The ${\rm M}^{\natural}$-concavity of $f_2$ follows from Proposition~\ref{prop:Murota}. 
 It is easy to see that $f_2$ satisfies \eqref{cond:min}, in a similar way as Lemma~\ref{prop:simple}. 
\end{proof}

%
%
%

 Lemma~\ref{teiri1} is immediate from Lemma~\ref{lem:f1=f2}. 
 Each of Lemmas \ref{teiri1} and \ref{lem:f1=f2} implies Theorem~\ref{thm:main}.

\section{The indicator function of a distributive lattice}
 Let $r \in \mathbb{R}_{>0}$ be an arbitrary. 
 As given a finite poset ${\cal P}=(N,\preccurlyeq)$, 
  we define $g_r \colon 2^N \to \mathbb{R}$ by 
\begin{eqnarray}
 g_r(X) = \exp(-r f(X)) 
\end{eqnarray}
 for $X \in 2^N$, where $f$ is an ${\rm M}^{\natural}$-concave set function satisfying \eqref{cond:min}. 
 For convenience, let $g_{\infty}(X) = \lim_{r \to \infty}g_r(X)$, then 
\begin{eqnarray}
g_{\infty}(X) = \begin{cases}
1 & \mbox{if } X \in {\cal I}({\cal P}), \\
0 & \mbox{otherwise}, 
\end{cases}
\end{eqnarray}
 holds for any $X \in 2^N$, 
  meaning that $g_{\infty}$ is the indicator function of the distributed lattice ${\cal I}({\cal P})$. 

 Specifically, 
  if we set $r=(n+2)\ln 2$, then we obtain $g_r(X) = 2^{-(n+2) f(X)}$, and then 
  $|{\cal I}({\cal P})| - \frac{1}{4} 
   \leq \sum_{X \in 2^N} g_r(X) 
   \leq |{\cal I}({\cal P})| + \frac{1}{4}$ holds. 
 This implies that 
  if we have an approximation algorithm 
    for the partition function of a log-${\rm M}^{\natural}$-convex function, 
  then we can approximate the number of ideals of a poset. 
 By standard arguments (cf.~\cite{JVV86,DGGJ04}) 
  about a {\em fully polynomial-time randomized approximation scheme} ({\em FPRAS}), 
  we can conclude as follows. 
\begin{theorem}
 If there is 
  a polynomial time approximate sampler for a log-${\rm M}^{\natural}$-convex distribution, 
  then {\em counting bipartite independent set} ({\em {\#}BIS}) has an FPRAS. 
\end{theorem}

\section{Concluding Remarks}
 We have shown that 
  any finite distributive lattice is isomorphic to 
  the minimizer set of an ${\rm M}^{\natural}$-concave set function. 
 The result implies that sampling from log-${\rm M}^{\natural}$-convex set function is {\#}BIS-hard 
   under the polynomial-time randomized approximate reduction. 
 It is a major open problem if an FPRAS exists for {\#}BIS (cf.~\cite{CGGGJSV16}), 
   with some applications 
   such as stable matching~\cite{Cheng08,KN12}.

\section*{Acknowledgments}
 The authors are grateful to Kazuo Murota 
   for his suggestion about the problem, and 
   for his kind advice on our preliminary manuscript. 
 The authors also grateful to Satoru Fujishige, Akihisa Tamura and Naoyuki Kamiyama for their valuable comments. 
 This work is partly supported by JST PRESTO Grant Number JPMJPR16E4, Japan.

\appendix
\section{Proof of Proposition~\ref{prop:Murota}}
\begin{proof}[Proof of Proposition~\ref{prop:Murota} (cf. \cite{MurotaB00} (Example 5.2.4, p.\ 282))] 
 Suppose $f(X) > -\infty$ and $f(Y) > -\infty$ for $X,Y \subseteq U$. 
 Let $M_X$ and $M_Y$ be respectively matchings attaining $f(X)$ and $f(Y)$, 
  i.e., $f(X)=w(M_X)$ and $f(Y)=w(M_Y)$. 
 Let $M  = M_X \cup (P \cap M_Y) \setminus (P \cap M_X)$ and 
  let $M' = M_Y \cup (P \cap M_X) \setminus (P \cap M_Y)$, where  
  $P \subseteq M_X \cup M_Y$ is the alternating path from $u \in X \setminus Y$. 
 Let $v \in U \cup V$ denote the other end of $P$. 
 In case of $v \in V$, 
   $\partial M \cap U = X - u$ and $\partial M' \cap U = Y+u$ hold. 
 This implies that $f(X-u) + f(Y+u) \geq w(M)+w(M') = w(M_X)+w(M_Y) = f(X)+f(Y)$, and 
  we obtain \eqref{mconcave1}. 
 In the other case, i.e., $v \in U$, 
   $\partial M \cap U = X - u+v$ and $\partial M' \cap U = Y+u-v$ hold. 
 This implies that $f(X-u+v) + f(Y+u-v) \geq w(M)+w(M') = w(M_X)+w(M_Y) = f(X)+f(Y)$, and 
  we obtain \eqref{mconcave2}. 
\end{proof}

\end{document}